\documentclass[twocolumn,10pt,letter]{IEEEtran}
\usepackage{longtable}
\usepackage{color}
\usepackage{soul}
\usepackage{amsmath}
\usepackage{cancel}
\usepackage{amssymb}
\usepackage{graphicx}
\usepackage{longtable}
\usepackage{multirow}
\usepackage{stfloats}
\usepackage{array}
\usepackage[labelsep=period]{caption}
\usepackage{setspace}
\usepackage{float}
\usepackage{subfig}
\usepackage{lettrine} 
\usepackage[noadjust]{cite}
\usepackage{balance}
\usepackage[english]{babel}
\usepackage[figurename=Fig.]{caption}
\usepackage{suffix}
\usepackage{mathtools}
\usepackage{url} 
\usepackage{cleveref}
\usepackage{soulutf8}

\newtheorem{proposition}{Proposition}

\newenvironment{proof}{{\textit{Proof:}}}{\hfill$\blacksquare$}
\newcommand{\qed}{\hfill $\blacksquare$}

\begin{document} 
\title{Hardware and Interference Limited Cooperative CR-NOMA Networks under Imperfect SIC and CSI}	
\author{Sultangali~Arzykulov,~\IEEEmembership{Member,~IEEE,} ~Galymzhan~Nauryzbayev,~\IEEEmembership{Member,~IEEE,}\\
~Abdulkadir Celik,~\IEEEmembership{Senior~Member,~IEEE,} and
Ahmed M. Eltawil,~\IEEEmembership{Senior~Member,~IEEE}
}
	\maketitle
	\thispagestyle{empty}
	\begin{abstract}
		The conflation of cognitive radio (CR) and non-orthogonal multiple access (NOMA) concepts is a promising approach to fulfill the massive connectivity goals of future networks given the spectrum scarcity. Accordingly, this letter investigates the outage performance of imperfect cooperative CR-NOMA networks under hardware impairments and interference. Our analysis is involved with the derivation of the end-to-end outage probability (OP) for secondary NOMA users by accounting for imperfect channel state information (CSI), as well as the residual interference caused by successive interference cancellation (SIC) errors and coexisting primary/secondary users. The numerical results validated by Monte Carlo simulations show that CR-NOMA network provides a superior outage performance over orthogonal multiple access. As imperfections become more significant, CR-NOMA is observed to deliver relatively poor outage performance.  
		\end{abstract}
	\IEEEpeerreviewmaketitle
	\section{Introduction}
\IEEEPARstart{T}{he} ambitious quality-of-service (QoS) demands of future wireless networks poses daunting challenges, especially under the ever-increasing number of devices connected to the Internet \cite{IoT_forecast}. This consequently leads to two problems: spectrum scarcity and interference-limited networks. Spectrum scarcity has been mostly studied in the realm of cognitive radio (CR) networks where unlicensed/secondary users are permitted to operate on spectrum bands licensed to primary users in an opportunistic and non-intrusive manner \cite{Celik2016Green}. Alternatively, non-orthogonal multiple access (NOMA) schemes have also recently received attention as a promising technique to mitigate the inability of orthogonal multiple access (OMA) schemes to support massive connectivity \cite{Ding_assump}. On the transmitter side, the NOMA scheme superposes messages of intended users with distinctive power allocation (PA) weights based on their channel quality. On the receiver side, the intended signal is extracted by decoding the transmitted broadcast message using successive interference cancellation (SIC). 

Therefore, the conflation of CR and NOMA concepts (CR-NOMA) is regarded as a potential solution to the problems mentioned above.  Existing literature on cooperative CR-NOMA networks mostly deals with simple scenarios under ideal cases without paying sufficient attention to practical limitations in terms of channel and hardware impairments \cite{Yu_CR_NOMA,Wang_CR_NOMA,Jia_CR_NOMA,Kumar_CR_NOMA}.  Accordingly, this letter investigates the outage performance of a generic and imperfect CR-NOMA, where non-ideality is modeled by accounting for hardware impairments (HIs), channel state information (CSI) imperfections, and residual interference due to the SIC error propagation. Ensuring that primary traffic is protected, our analysis also tackles the random interference caused by coexisting primary/secondary users. The numerical results validated by Monte Carlo simulations show that the CR-NOMA network provides a superior outage performance over orthogonal multiple access, which degrades with CSI and SIC imperfections.
\section{System Model}
\label{sec:system model}
\begin{figure}[t]
	\centering
	\includegraphics[width=1\columnwidth]{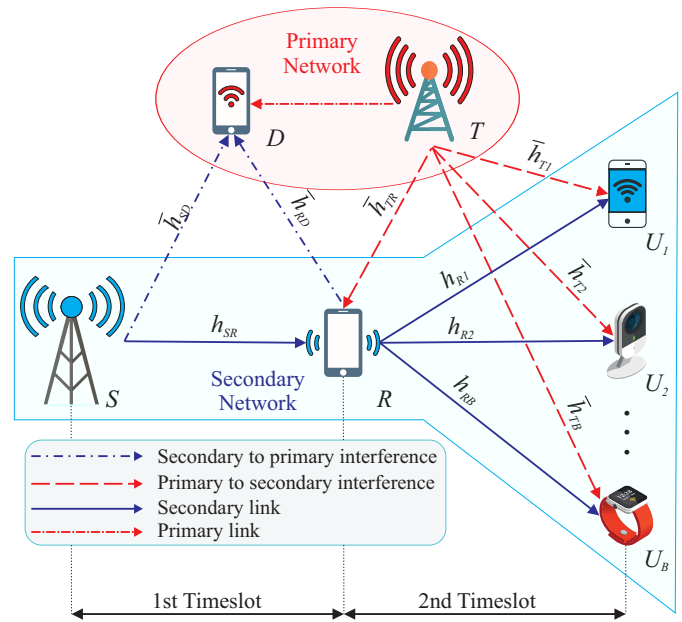}
	\caption{The proposed underlay CR-NOMA network.}
	\label{system_model}
\end{figure}
\subsection{Channel Model}
We consider a downlink CR-NOMA network that consists of primary and secondary networks as illustrated in Fig. \ref{system_model}. The primary network comprises of primary transmitter ($T$) and destination ($D$) nodes. On the other hand, the secondary network consists of a source ($S$),  a relay ($R$) that cooperates with $S$ in a half-duplex decode-and-forward (DF) mode, and $B$ secondary NOMA users. The cooperation occurs in two time slots; The broadcast signal transmitted by $S$ in the first time slot is re-transmitted to $B$ NOMA users in the second time slot. Channel gains among the nodes are modeled as $\bar{h}_i\sim \mathcal{CN}(0,1),~\forall i \in \{SD, SR, RD, R1, \ldots, RB, TR, T1, \ldots, TB\}$. Moreover, the distance between the corresponding nodes and the path-loss exponent are denoted by $d_i$ and $\tau$, respectively. To capture CSI imperfections, we model channel coefficients  using the minimum mean square error (MMSE) channel estimator as $\bar{h}_i = h_i + e_i$,  where $h_i \sim \mathcal{CN}(0,\sigma_{h_i}^2)$ and $e_i \sim  \mathcal{CN}(0,\zeta_i)$ are the estimated channel coefficient and channel estimation error with variance $\sigma_{h_i}^2$ and $\zeta_i$, respectively. The error variance is modeled as $\zeta_i \triangleq \theta \rho^{- \kappa}$ where $\rho = \frac{P}{\sigma^2}$ is the transmitted signal-to-noise ratio (SNR), and $\kappa \geqslant 0 $, $\theta > 0$  \cite{impCSI}. Indeed, this model describes various  CSI acquisition scenarios: a) $\zeta$ is a function of $\rho$ for $\kappa \neq 0$, and b) $\zeta$ is independent of $\rho$ for $\kappa = 0$. Following from the underlay CR paradigm, it is also assumed that the transmit power of a secondary node $j \in \{R,S\}$ is restricted as
\cite{Arzykulov_TCCN} $P_j  \leqslant \min\left(\bar{P}_j, \frac{I_{\text{ITC}, j}}{|\bar{h}^{}_{j D}|^2}\right)$, 
where $\bar{P}_j$ stands for the maximum transmit power at node $j$ and $I_{\text{ITC},j}$ denotes the interference temperature constraint (ITC) at $D$ caused by node $j$. 
\subsection{Transmission Protocol}
In the first time slot of the proposed CR-NOMA relaying model, $S$ broadcasts $\sum_{b=1}^{B}\sqrt{\alpha_b} x_b$ to $R$, where $\alpha_b$ is the PA factor\footnote{Similar to \cite{Ding_QoS}, we allocate the power based on the QoS requirements imposed at secondary NOMA users.} of $U_b$ such that $\alpha_1 >  \ldots > \alpha_b > \ldots > \alpha_B$ and $\sum_{b=1}^{B} \alpha_b = 1$, and $x_b$ is the message dedicated to $U_b$ with $\mathbb{E}(|x_b|^2)=1$. 

Considering CSI imperfections and aggregate distortion noise, the received signal at $R$ can be written as 
\begin{align}
\label{signal at R}
y_{R} =& \left( h^{}_{SR} + e^{}_{SR}\right) \sqrt{\tilde{P}^{}_S} \left( \sum_{b=1}^{B} \sqrt{\alpha^{}_b } x^{}_b +  \eta^{}_{SR} \right) \nonumber\\
&+   \bar{h}^{}_{TR} \sqrt{\tilde{P}^{}_T} \left(x^{}_T + \eta^{}_T\right) + n^{}_R,  
\end{align}
where $\tilde{P}^{}_i = \frac{P_i}{d_{iR}^{\tau}}$, $i \in \{S,T\}$; 
$\eta^{}_{(\cdot)} \sim \mathcal{CN}\left(0, \phi^{2}_{(\cdot)} \right)$ denotes the aggregate distortion noise from transceiver; $\phi^{}_{(\cdot)}$ = $\sqrt{\phi^{2}_{t}+\phi^{2}_{r}}$ is the aggregate HI level from the transmitter and receiver \cite{Stefania};  $n_{(\cdot)} \sim \mathcal{CN}(0,\sigma_{(\cdot)}^2)$ denotes the additive white Gaussian noise (AWGN) term at each receiver node; $P_T$ and $x^{}_T$ stand for the transmit power at $T$ and the message dedicated to $D$, respectively. Then, the instantaneous signal-to-interference-distortion-noise-ratio (SIDNR) to decode $x_j$,  $1\leq j \leq b <  B$, at $R$ can be expressed by
\begin{align}
	\label{gammaRb_SINR}
	\gamma^{}_{R,j} &=\frac{\alpha^{}_j P_S \left| h^{}_{SR}\right| ^2  }{  \mathcal{A} P_S  \left|h^{}_{SR}\right|^2  + \mathcal{C} P_S + \mathcal{D}_R \tilde{P}_T \left|\bar{h}^{}_{T R}\right|^2 + d^{\tau}_{SR} \sigma^2_{R} },
\end{align} 
where $\mathcal{A} = \Omega_j+ \tilde{\Omega}_j  + \phi^{2}_{SR}$;  $\Omega_j = \sum_{n=j+1}^{B} \alpha^{}_n$;  $\tilde{\Omega}_j = \sum_{\tilde{n}=1}^{j-1} \backepsilon^{}_{\tilde{n}} \alpha^{}_{\tilde{n}}$, $0<\backepsilon^{}_{\tilde{n}}<1$, where $\backepsilon^{}_{\tilde{n}}=0$ and $\backepsilon^{}_{\tilde{n}}=1$ indicate the perfection and absence of SIC, respectively; $\mathcal{C} = \zeta^{}_{SR} + \zeta^{}_{SR} \phi^{2}_{SR}$ and $\mathcal{D}_R = d^{\tau}_{SR} d^{-\tau}_{TR}\left(1+ \phi^{2}_{TR}\right) $. Furthermore, by assuming imperfect detection of $x^{}_j$, $R$ decodes the message of user $B$ with the SIDNR of
\begin{align}
	\label{gammaRB_SINR}
	\hspace{-0.2cm}\gamma^{}_{R,B} &=\frac{\alpha^{}_B P_S \left|h^{}_{SR}\right|^2 }{ \mathcal{\tilde{A}}  P_S \left|h^{}_{SR}\right|^2 + \mathcal{C} P_S + \mathcal{D}_R \tilde{P}_T  \left|\bar{h}^{}_{T R}\right|^2 + d^{\tau}_{SR} \sigma^2_{R} },
\end{align} 
where $\mathcal{\tilde{A}}=\left( \tilde{\Omega}_B + \phi^{2}_{SR}\right)$ and $\tilde{\Omega}_B = \sum_{\tilde{n}=1}^{B-1} \backepsilon^{}_{\tilde{n}} \alpha^{}_{\tilde{n}}$.

In the second time slot, $R$ relays the decoded signal $\sum_{b=1}^{B}\sqrt{\beta^{}_b} \tilde{x}^{}_b$ to $B$ NOMA users, where $\beta^{}_b$, with $\sum_{b=1}^{B} \beta^{}_b = 1$, is the PA factor of $U_b$. Hence, the received signal at $U_b$ can be written as
\begin{align}
\label{signal at b}
y^{}_{b} &=  \left( h^{}_{R b} + e^{}_{b}\right) \sqrt{\tilde{P}^{}_R} \left( \sum_{j=1}^{B} \sqrt{\beta^{}_j } x^{}_j +  \eta^{}_{b} \right) \nonumber\\
&~~~+   \bar{h}^{}_{Tb} \sqrt{\tilde{P}^{}_T} \left(x^{}_T + \eta^{}_T\right) + n^{}_b,  
\end{align}
where $\tilde{P}^{}_i = \frac{P_i}{d^{\tau}_{ib}}$, $i \in \{R,T\}$.
Then, we can write the SIDNR for $U_b$ to detect the message of $U_j$ as follows
\begin{align}
\label{gammab_SNDIR}
\gamma^{}_{b,j} &=\frac{\beta^{}_j P_{R} \left| h^{}_{R b}\right| ^2  }{\mathcal{J}_b  P_{R} \left|h^{}_{R b}\right|^2  + \mathcal{G}_b P_R + \mathcal{D}_b \tilde{P}_{T}  \left|\bar{h}^{}_{Tb}\right|^2 + d^{\tau}_{Rb} \sigma^2_{b} },
\end{align} 
where $\beta^{}_1>\beta^{}_{j}>\beta^{}_{b}>\beta^{}_{B}$; $\forall b \in \{1, \ldots, B\}$;  $\mathcal{J}_b = \left( \Omega_j + \tilde{\Omega}_j + \phi^{2}_{b} \right)$; $\Omega_j = \sum_{n=j+1}^{b} \beta^{}_n$; $\tilde{\Omega}_j = \sum_{\tilde{n}=1}^{j-1} \backepsilon_{\tilde{n}} \beta^{}_{\tilde{n}}$; $\mathcal{G}_b = \left(\zeta^{}_b + \zeta^{}_b \phi^{2}_{b}  \right) $ and $\mathcal{D}_b = d^{\tau}_{Rb}d^{-\tau}_{Tb}\left(1+ \phi^{2}_{Tb}\right) $. Finally, after decoding messages of $B-1$ NOMA users, $U_B$ detects its own message with
\begin{align}
\label{gammaB}
\hspace{-0.2cm}\gamma^{}_{B} &=\frac{\beta^{}_B P_{R} \left|h^{}_{R B}\right|^2 }{\mathcal{J}_B  P_{R} \left|h^{}_{R B}\right|^2 + \mathcal{G}_B P_R + \mathcal{D}_B \tilde{P}_T \left|\bar{h}^{}_{T B}\right|^2 + d^{\tau}_{RB} \sigma^2_{B} },
\end{align} 
where $\mathcal{J}_B = \tilde{\Omega}_B + \phi^{2}_{B}$; $\tilde{\Omega}_B = \sum_{\tilde{n}=1}^{B-1} \backepsilon_{\tilde{n}} \beta^{}_{\tilde{n}}$. 
Considering the dual-hop communication, the achievable rate at $U_j$ is calculated as 
\begin{align}
	\label{rate}
	\mathcal{R}_{j} = \frac{1}{2}\log_2\left[1 + \min\left( \gamma^{}_{R,j}, \gamma^{}_{b,j}\right)\right],~1\leq j \leq b \leq  B.
\end{align}
\section{Outage Analysis}
\label{sec:Outage Analysis}
The outage at $U_j$ occurs when the achievable rate at $U_j$ is below a predefined rate threshold $\mathcal{R}_{{\text{th}},j}$. Following from \eqref{rate}, the OP of $U_j$ can be written as
\begin{align}
\label{Poutj}
	P_{{\text{out}},j} &= {\text{Pr}}\left[ \frac{1}{2} \log_2 \left[1+ \min\left(\gamma^{}_{R,j}, \gamma^{}_{b,j} \right)  \right]  < \mathcal{R}_{{\text{th}},j}\right] \nonumber\\
	&= 1 - {\text{Pr}}\left[\min\left(\gamma^{}_{R,j}, \gamma^{}_{b,j} \right) > \psi_j \right] \nonumber\\
	&= F_{\gamma^{}_{R,j}}(\psi_j) + F_{\gamma^{}_{j}}(\psi_j) - F_{\gamma^{}_{R,j}}(\psi_j) F_{\gamma^{}_{j}}(\psi_j),
\end{align} 
where $\psi_j = 2^{2\mathcal{R}_{{\text{th}},j}} - 1$ denotes the predefined SNR threshold at $U_j$. Now, considering the ITC imposed at $D$, we can write the CDF of the RV $\gamma^{}_{R,j}$ as in \eqref{CDF_gammaRj}, where $X = \left| h^{}_{SR}\right| ^2$, $Y = \left| \bar{h}^{}_{SD}\right| ^2$ and $Z=\left| \bar{h}^{}_{PR}\right| ^2$ follow the exponential distribution with parameter $\lambda$; $\mathcal{K}_c = \frac{\mathcal{D}_R \tilde{P}_T \psi_j}{ d \left( \alpha_j - \mathcal{A} \psi_j\right)  }$; $\mathcal{M}_c = \frac{\psi_j d^{\tau}_{SR} \sigma^2_{j} }{d \left( \alpha_j - \mathcal{A} \psi_j\right) }$, with $c=d$,  $c \in \{\Delta, \Upsilon\}$ and $d \in \{\bar{P}_S, I^{}_{\text{ITC}} d^{\tau}_{SD} \}$; $\mathcal{L} = \frac{\mathcal{C}  \psi_{j} }{\alpha_j - \mathcal{A} \psi_j}$ and $\Lambda_S = \frac{I^{}_{\text{ITC}} d^{\tau}_{SD}}{\bar{P}_S}$.
\begin{figure*}[t!]
	\begin{align}
	\label{CDF_gammaRj}
	F_{\gamma_{R,j}} (\psi_j) &= {\text{Pr}}\left[ \frac{\alpha^{}_j \bar{P}_S \left| h^{}_{SR}\right| ^2  }{  \mathcal{A} \bar{P}_S  \left|h^{}_{SR}\right|^2  + \mathcal{C} \bar{P}_S + \mathcal{D}_R \tilde{P}_T \left|\bar{h}^{}_{T R}\right|^2 + d^{\tau}_{SR} \sigma^2_{R} }<\psi_j, ~ \bar{P}_S < \frac{I^{}_{\text{ITC}}d^{\tau}_{SD}}{|\bar{h}^{}_{S D}|^2} \right]\nonumber\\ 
	&~~~+ {\text{Pr}}\left[ \frac{ \frac{\alpha^{}_j I^{}_{\text{ITC}} d^{\tau}_{SD} \left| h^{}_{SR}\right| ^2}{|\bar{h}^{}{S D}|^2}   }{   \frac{\mathcal{A} I^{}_{\text{ITC}} d^{\tau}_{SD} \left| h^{}_{SR}\right| ^2}{|\bar{h}^{}{S D}|^2}  +  \frac{\mathcal{C} I^{}_{\text{ITC}} d^{\tau}_{SD}}{|\bar{h}^{}{S D}|^2} + \mathcal{D}_R \tilde{P}_T \left|\bar{h}^{}_{T R}\right|^2 + d^{\tau}_{SR} \sigma^2_{R} }<\psi_j,~ \bar{P}_S > \frac{I^{}_{\text{ITC}} d^{\tau}_{SD} }{|\bar{h}^{}_{S D}|^2} \right] \nonumber\\
	&= \underbrace{{\text{Pr}}\left[ X<Z \mathcal{K}_\Delta + \mathcal{M}_\Delta + \mathcal{L}, Y<\Lambda_S \right]}_{\text{$\Delta$}}+  \underbrace{{\text{Pr}}\left[ X<ZY \mathcal{K}_\Upsilon+Y \mathcal{M}_\Upsilon + \mathcal{L}, Y>\Lambda_S \right]}_{\text{$\Upsilon$}},
	\end{align}
	\hrulefill
\end{figure*}  
\begin{proposition}
	\label{proposition1}
	The CDF of $\gamma_{R,j}$ can be derived in its closed-form as
	\begin{align}
	\label{CDF_gammaRj_final}
F_{\gamma_{R,j}}& (\psi_j) =	1 - \frac{\lambda_z e^{-\lambda_x \left(  \mathcal{M}_\Delta + \mathcal{L} \right) }}{\lambda_z + \lambda_x \mathcal{K}_\Delta }  \left(1 - e^{-\lambda_y \Lambda_S} \right)\nonumber\\
& + \frac{\lambda_y \lambda_z \text{Ei} \left[ -\mu^{}_S \xi^{}_S\right]      }{\lambda_x \mathcal{K}_\Upsilon} e^{-\Lambda_S \left(\lambda_y + \lambda_x \mathcal{M}_\Upsilon \right) -\lambda_x \mathcal{L}  + \mu^{}_S \xi^{}_S }. 
	\end{align}
\end{proposition}~~~\begin{proof}
	See Appendix \ref{Appendix 1}.
\end{proof}

Following from \eqref{gammab_SNDIR}, the CDF of $\gamma^{}_{b,j}$ can be described as 
\begin{align}
\label{CDF_gammabj}
 F_{\gamma^{}_{b,j}} (\psi_j) &= \underbrace{{\text{Pr}}\left[ Q<W \mathcal{S}_\Theta + \mathcal{O}_\Theta + \mathcal{T}, V<\Lambda_R \right]}_{\text{$\Theta$}}\nonumber\\
 &+  \underbrace{{\text{Pr}}\left[ Q<W V \mathcal{S}_\Phi + V \mathcal{O}_\Phi + \mathcal{T}, V>\Lambda_R \right]}_{\text{$\Phi$}},
\end{align}
where $Q = |h_{R b}|^2$; $V = |\bar{h}^{}_{R D}|^2$; $W = |\bar{h}^{}_{T b}|^2$; $\mathcal{S}_k = \frac{\mathcal{D}_b \tilde{P}_T \psi_j}{l \left( \beta^{}_b - \mathcal{J}_b \psi_j\right) }$; $\mathcal{O}_k = \frac{d^{\tau}_{Rb}\sigma^2_b \psi_j}{l \left( \beta^{}_b - \mathcal{J}_b \psi_j\right) }$, with $k=l$, $k \in \{\Theta, \Phi\}$ and $l \in \{\bar{P}^{}_R, I^{}_{\text{ITC}} d^{\tau}_{RD}\}$; $\mathcal{T} = \frac{\mathcal{G}_b \psi_j}{\beta^{}_b - \mathcal{J}_b \psi_j}$ and $\Lambda_R = \frac{I^{}_{\text{ITC}} d^{\tau}_{RD}}{\bar{P}^{}_R}$.
\begin{proposition}
	\label{proposition2}
The closed-from expression for the CDF of $\gamma^{}_{b,j}$ can be written as
\begin{align}
\label{CDF_gammabj_final}
F_{\gamma^{}_{b,j}}& (\psi_j) =1 - \frac{\lambda_w e^{-\lambda_q \left(  \mathcal{O}_\Theta + \mathcal{T} \right) }}{\lambda_w + \lambda_q \mathcal{S}_\Theta }  \left(1 - e^{-\lambda_v \Lambda_R} \right)\nonumber\\
& + \frac{\lambda_v \lambda_w \text{Ei} \left[ -\mu^{}_R \xi^{}_R\right]      }{\lambda_q \mathcal{S}_\Phi} e^{-\Lambda_R \left(\lambda_v + \lambda_q \mathcal{O}_\Phi \right) -\lambda_q \mathcal{T}  + \mu^{}_R \xi^{}_R }. 
\end{align}
\end{proposition}
\begin{proof}
See Appendix \ref{Appendix 2}.
\end{proof}

Lastly, the exact OP of $U_j$ can be derived after substituting  \eqref{CDF_gammaRj_final} and \eqref{CDF_gammabj_final} into \eqref{Poutj}. 
Similarly, the OP of the $U_M$ can be derived by using \eqref{gammaRb_SINR} and \eqref{gammaB} and following the same procedure in obtaining the OP of $U_j$. Then, following the similar approach as in Appendix \ref{Appendix 1}, the CDFs of the RV $\gamma^{}_{R,B}$ can be derived in closed-form as
\begin{align}
\label{CDF_gammaRB_final}
F_{\gamma_{R,B}}& (\psi_B) =	1 - \frac{\lambda_z e^{-\lambda_x \left(  \mathcal{H}_\Psi + \mathcal{I} \right) }}{\lambda_z + \lambda_x \mathcal{E}_\Psi }  \left(1 - e^{-\lambda_y \Lambda_S} \right)\nonumber\\
& + \frac{\lambda_y \lambda_z \text{Ei} \left[ -\bar{\mu}^{}_S \bar{\xi}^{}_S\right]      }{\lambda_x \mathcal{E}_\Pi} e^{-\Lambda_S \left(\lambda_y + \lambda_x \mathcal{H}_\Pi \right) -\lambda_x \mathcal{I}  + \bar{\mu}^{}_S \bar{\xi}^{}_S }, 
\end{align}
where $\mathcal{E}_m = \frac{\mathcal{D}_R \tilde{P}_T \psi_B}{ d \left( \alpha_B - \mathcal{\tilde{A}} \psi_B\right)  }$; $\mathcal{H}_m = \frac{ d^{\tau}_{SR} \sigma^2_{B} \psi_B }{d \left( \alpha_B - \mathcal{\tilde{A}} \psi_B\right) }$, with $m=d$,  $m \in \{\Psi, \Pi \}$ and $d \in \{\bar{P}_S, I^{}_{\text{ITC}} d^{\tau}_{SD} \}$; $\mathcal{I} = \frac{\mathcal{C \psi_B}}{\alpha_B - \mathcal{\tilde{A}} \psi_j}$; $\bar{\mu}^{}_S = \lambda_z +\lambda_x \Lambda_S \mathcal{E}^{}_\Pi$;  $\bar{\xi}^{}_S = \frac{ \mathcal{H}^{}_\Pi}{ \mathcal{E}^{}_\Pi} + \frac{\lambda_y}{\lambda_x \mathcal{E}^{}_\Pi}$. Moreover, $\psi^{}_B <\frac{\alpha^{}_B}{\mathcal{\tilde{A}}}$, otherwise, $F_{\gamma^{}_{R,B}} (\psi^{}_B) \sim 1$.
Similarly, the CDF of $\gamma^{}_{B}$ is derived as
\begin{align}
\label{CDF_gammaB_final}
F_{\gamma^{}_{B}}& (\psi_B) =1 - \frac{\lambda_{\bar{w}} e^{-\lambda_q \left(  \mathcal{\bar{O}}_\Xi + \mathcal{\bar{T}} \right) }}{\lambda_{\bar{w}} + \lambda_{\bar{q}} \mathcal{\bar{S}}_\Xi }  \left(1 - e^{-\lambda_{\bar{v}} \Lambda_R} \right)\nonumber\\
& + \frac{\lambda_{\bar{v}} \lambda_{\bar{w}} \text{Ei} \left[ -\mu^{}_R \xi^{}_R\right]      }{\lambda_{\bar{q}} \mathcal{\bar{S}}_\aleph} e^{-\Lambda_R \left(\lambda_{\bar{v}} + \lambda_{\bar{q}} \mathcal{\bar{O}}_\aleph \right) -\lambda_{\bar{q}} \mathcal{\bar{T}}  + \bar{\mu}^{}_R \bar{\xi}^{}_R }, 
\end{align}
where $\bar{Q} = |h_{R B}|^2$; $\bar{W} = |\bar{h}^{}_{T B}|^2$; $\bar{\mathcal{S}}_r = \frac{\mathcal{\bar{D}}_B \tilde{P}_T \psi_B}{n \left( \beta^{}_B - \mathcal{\bar{J}}_B \psi_B \right) }$; $\mathcal{\bar{O}}_r = \frac{d^{\tau}_{RB}\sigma^2_B \psi_B}{n \left( \beta^{}_B - \mathcal{\bar{J}}_B \psi_B\right) }$, with $r=n$, $r \in \{\Xi, \aleph\}$ and $n \in \{\bar{P}^{}_R, I^{}_{\text{ITC}} d^{\tau}_{RD}\}$; $\mathcal{\bar{T}} = \frac{\mathcal{\bar{G}}_B \psi_B}{\beta^{}_B - \mathcal{\bar{J}}_B \psi_B}$;  	$\bar{\mu}^{}_R = \lambda_{\bar{w}} +\lambda_{\bar{q}} \Lambda_R \mathcal{\bar{S}}^{}_\aleph$ and $\bar{\xi}^{}_R = \frac{ \mathcal{\bar{O}}^{}_\aleph}{ \mathcal{\bar{S}}^{}_\aleph} + \frac{\lambda_{\bar{v}}}{\lambda_{\bar{q}} \mathcal{\bar{S}}^{}_\aleph}$. Notice that $\psi^{}_B <\frac{\beta^{}_B}{\mathcal{\bar{J}}_B}$, otherwise, $F_{\gamma^{}_{B}} (\psi^{}_B) \sim 1$. Finally, the exact OP of $U_B$ can be derived by using  \eqref{CDF_gammaRB_final} and \eqref{CDF_gammaB_final}. 
\begin{figure}[!t]
	\centering
	\includegraphics[width=1\columnwidth]{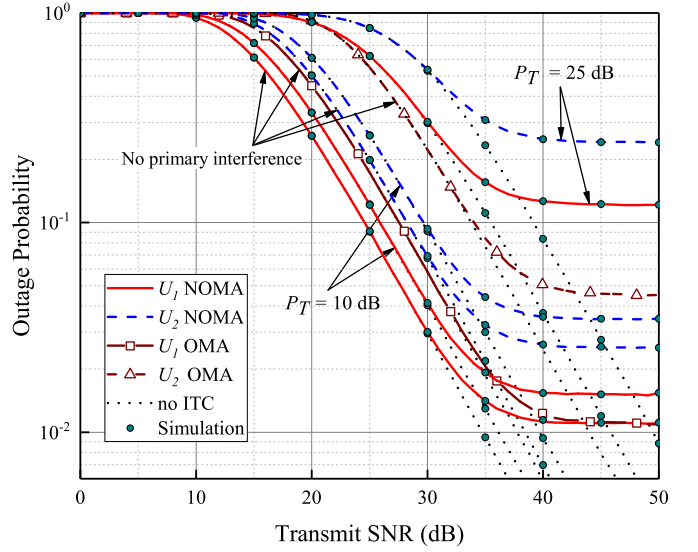}
	\caption{The OP versus the transmit SNR for NOMA and OMA users with $I_{\text{ITC}} = 20$ dB, $\phi = 0$, $\backepsilon = 0$ and $\zeta = 0$.}
	\label{fig2}
\end{figure}
\section{Numerical Results}
\label{sec:numerical results}
This section discusses the numerical results and validates that all theoretical analyses precisely match with Monte-Carlo simulations. We assume two secondary NOMA users\footnote{In practice, it may not be feasible to consider many NOMA users due to the complexity and latency of SIC receivers, which increases non-linearly with the increase in the number of users \cite{Ding_assump}. In this work, SIC complexity especially becomes more significant because of the SIC error propagation.}, i.e.,  $U_1$ and $U_2$ with the following system settings; the same transmit power levels at $S$ and $R$, i.e., $P = P_S = P_R$; $\alpha_1 = \beta_1  = 0.8$; $\alpha_2 = \beta_2 = 0.2$; $\mathcal{R}_1 = 1$ bps;  $\mathcal{R}_2 = 1.5$ bps; $d_{S R}=d_{R 1}=d_{R 2}=d$; $d_{SD} = d_{RD} = d_{TR} = d_{T1} = d_{T2} = 3d$, where $d$ is assumed to be unity; $\tau = 3$. 

Fig. \ref{fig2} compares the OP of users operating on OMA and NOMA. For the sake of a fair comparison, the  QoS demands of cooperative OMA is set as two-fold of that used for cooperative NOMA. To demonstrate the ITC impacts on the OP of SUs, we consider the asymptotic case, where $D$ does not impose ITC, i.e., $I_{\text{ITC}} = \infty$. Fig. \ref{fig2} shows that $U_1$ achieves a lower OP than $U_2$ since $U_1$ is assigned with a lower rate and a higher PA factor. Also, notice that both NOMA users achieve a significantly lower OP than corresponding OMA modes. Additionally, we can observe that the ITC imposed at $D$, i.e., $I_{\text{ITC}}$, results in the saturation of OP curves. This implies that secondary transmitters cannot increase their transmission power above the ITC level in order not to cause harmful interference to $D$. A noteworthy observation is that an increase in primary interference level results in the outage performance deterioration of SUs. For example, at transmit SNR of $20$ dB, $U_1$ obtains the OP of $0.09$, $0.12$ and $0.62$ without primary interference, $P_T = 10$ dB and $P_T = 25$ dB, respectively. 

\begin{figure}[!t]
	\centering
	\includegraphics[width=1\columnwidth]{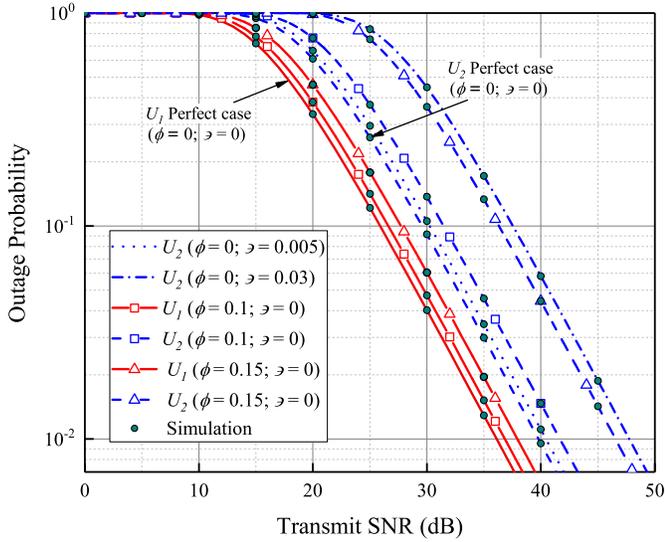}
	\caption{The OP versus the transmit SNR for NOMA users with $P_T = 10$ dB, $I_{\text{ITC}} = \infty$ and $\zeta = 0$.}
	\label{fig3}
\end{figure}

Fig. \ref{fig3} presents the impact of hardware and SIC imperfections on the OP of NOMA users considering perfect CSI ($\zeta = 0$). Here, we consider two imperfect SIC scenarios with $\backepsilon~= 0.005$ and $\backepsilon~= 0.03$. It is important to note that the higher level of SIC imperfection degrades the OP of NOMA users by causing full outage at intolerable imperfect SIC levels. For example, for the proposed system model with $\phi = 0$, the tolerable imperfect SIC level can be calculated from  $\backepsilon~< \frac{\alpha^{}_2 - \phi^2}{\alpha_1  \psi^{}_2}$ as $\backepsilon~< 0.035$. Therefore, the plot shows that the imperfect SIC degrades the outage performance of $U_2$. For instance, at $30$ dB transmit SNR, $U_2$ obtains the OP of $0.105$ and $0.45$ when $\backepsilon~= 0.005$ and $\backepsilon~= 0.03$, respectively, while the OP for perfect SIC is $0.09$. Additionally, we set two different HI levels as $\phi = 0.1$ and $\phi = 0.15$ to show the effect of HIs on the system performance. It is obvious that both NOMA users demonstrate better performance for the lower level of HI, as  expected. A specific observation is that $U_2$ is more sensitive to the distortion noise than $U_1$. For example, at $30$ dB transmit SNR and $\phi = 0.15$, the OP of $U_1$ and $U_2$ degrades for $0.025$ and $0.3$, accordingly. Moreover, after comparing the outage performance of NOMA and OMA users, we note that, even in hardware limited scenario, the NOMA model still outperforms the OMA one. In addition, it is also noticed that the impact of HIs is more effective on the OMA user. For instance, when $\phi = 0.15$, the OP of NOMA user 2 degrades for the value of about $0.3$, while the OP of OMA user 2 declares a full outage at all SNRs levels. 
\begin{figure}[!t]
	\centering
	\includegraphics[width=1\columnwidth]{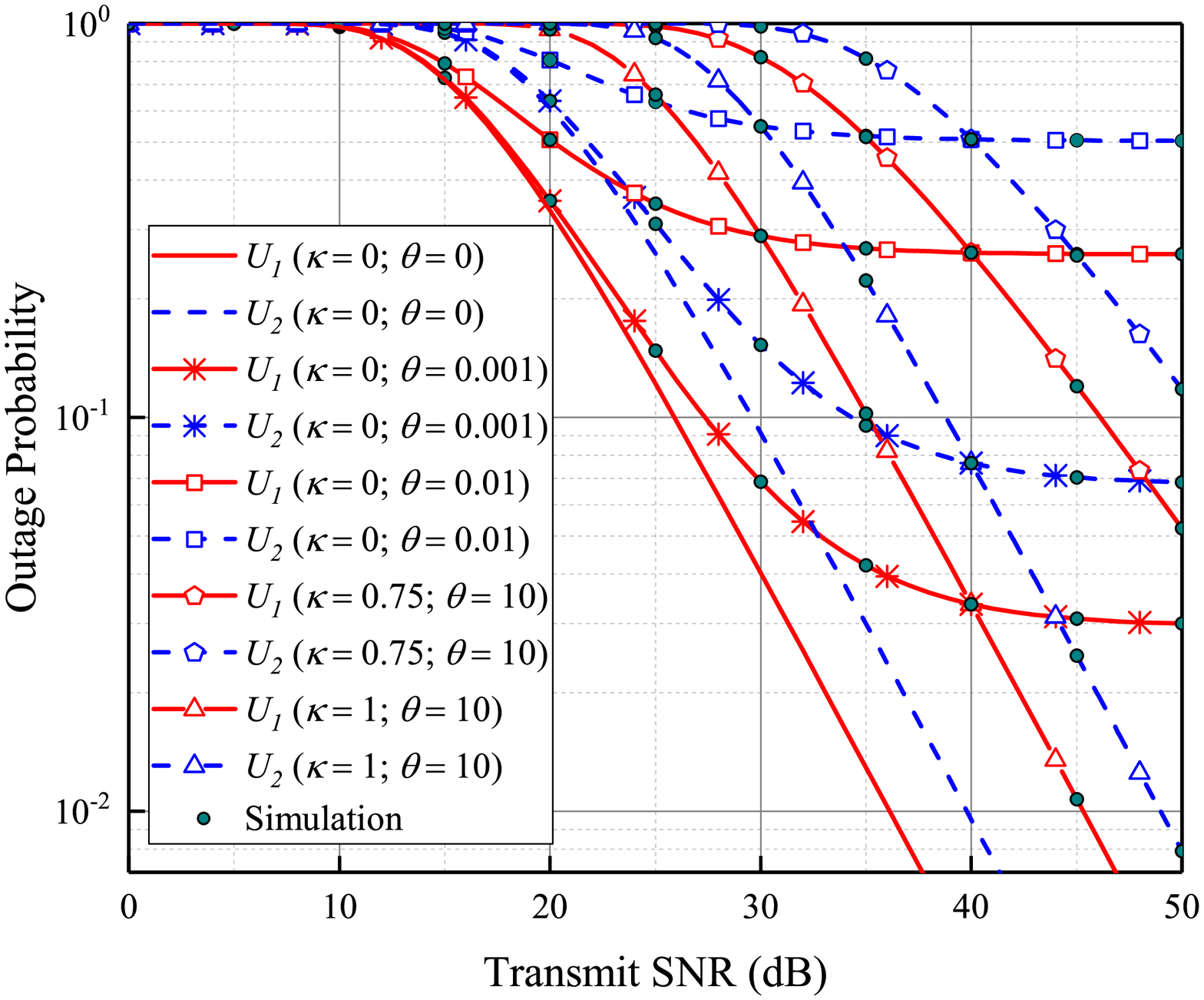}
	\caption{The OP versus the transmit SNR for NOMA users with $P_T = 10$ dB, $I_{\text{ITC}} = \infty$ and $\phi = \backepsilon~= 0$.}
	\label{fig4}
\end{figure}

In Fig. \ref{fig4}, for the sake of figures clarity, we show the impact of the channel error variance on the outage performance of only $U_1$ considering NOMA and OMA models.  We set the system parameters as $I_{\text{ITC}}  = \infty$, $P_T = 10$ dB, $\phi = 0$ and $\backepsilon~= 0$. It is observed from results that the NOMA user obtains better OP comparing with OMA one for all channel uncertainty scenarios. When $\kappa = 0$, the channel error variance becomes SNR-independent and increasing transmit SNR provides no advantage. However, the outage performance degrades by increasing  $\theta$. For instance, the OP of NOMA and OMA users saturate after $40$ dB and $30$ dB when $\theta = 0.001$ and $\theta = 0.01$, respectively. One observation is that, when $\theta$ is small, it does not cause considerable impact on the OP at lower SNR values. This reason is, when $\theta$ tends to zero, the channel estimation approaches to the perfect CSI. On the other hand, when $\theta = 0.1$, both NOMA and OMA users declare an outage at all SNR values, which means that the considered channel uncertainty is intolerable for these users. When $\kappa \neq 0$, OP saturation is not noticed as channel error model becomes SNR-dependent, and the increase of $\kappa$ results in an improvement of the outage performance as the channel error is inversely proportional to the SNR. For example, when $\kappa=1.5$ and $\theta=10$, we can see that the impact of channel error decreases by increasing the SNR level and outage curves of both NOMA and OMA users approach the performance the perfect CSI at high SNRs.

\section{Conclusion}
\label{sec:Conclusion}
This letter analyzed the performance of the downlink underlay CR-NOMA DF-based relaying network considering hardware, CSI, and SIC imperfections. Closed-form analytical expressions for the end-to-end OP of NOMA SUs were derived considering primary and secondary interference. Moreover, the proposed NOMA system model obtained better OP results compared to the OMA one, which is considered as a benchmark model. Finally, the accurateness of the derived analytical expressions were verified by Monte Carlo simulations. In the future, the considered system model can be extended by considering the device-to-device mmWave communication. 
\balance 
\appendices
\numberwithin{equation}{section}
	\section{Proof of Proposition 1}
	\label{Appendix 1}
The term $\Delta$ in \eqref{CDF_gammaRj} can be derived as
	\begin{align}
		\label{Delta}
	\Delta =& \int_{z=0}^{\infty} f_Z(z) \int_{x=0}^{z \mathcal{K}_\Delta + \mathcal{M}_\Delta + \mathcal{L}} f_X(x){\text d}x {\text d}z  \int_{y=0}^{\Lambda_S} f_Y(y){\text d}y\nonumber\\
	=& \left(1 - e^{-\lambda_y \Lambda_S} \right) \int_{z=0}^{\infty} \lambda_z e^{-\lambda_z z} \left(1 - e^{-\lambda_x\left(z \mathcal{K}_\Delta + \mathcal{M}_\Delta + \mathcal{L} \right)} \right)  \text{d}z\nonumber\\
	=& \left(1 - e^{-\lambda_y \Lambda_S} \right) \left(1- \frac{\lambda_z e^{-\lambda_x \left(\mathcal{M}_\Delta + \mathcal{L} \right) }}{\lambda_z + \lambda_x \mathcal{K}_\Delta } \right). 
	\end{align}
		
	Then, the term $\Upsilon$ in \eqref{CDF_gammaRj} can be rewritten as follows
	\begin{align}
		\label{Upsilon}
		\Upsilon&= \int_{z=0}^{\infty} f_Z(z) \underbrace{\int_{y=\Lambda_S}^{\infty} \int_{x=0}^{zy \mathcal{K}_\Upsilon + y \mathcal{M}_\Upsilon + \mathcal{L}} \hspace{-1.2cm} f_Y(y) f_X(x){\text d}x {\text d}y}_{\Upsilon_1}  {\text d} z,
	\end{align}
	while the term $\Upsilon_1$ in \eqref{Upsilon} can be calculated by
	\begin{align}
		\label{Upsilon_1}
		\Upsilon_1 &= \int_{y=\Lambda_S}^{\infty} \lambda_y e^{-\lambda_y y} \left(1 - e^{-\lambda_x \left(zy \mathcal{K}_\Upsilon + y \mathcal{M}_\Upsilon + \mathcal{L} \right)} \right)  \text{d}y\nonumber\\
		&= e^{-\lambda_y \Lambda_S} - \frac{\lambda_y e^{-\lambda_x \mathcal{L}}~ e^{-\Lambda_S \left(z \lambda_x  \mathcal{K}_\Upsilon + \lambda_x \mathcal{M}_\Delta+\lambda_y \right) } }{z \lambda_x  \mathcal{K}_\Upsilon + \lambda_x \mathcal{M}_\Delta+\lambda_y }.
	\end{align}
Then, by inserting \eqref{Upsilon_1} into \eqref{Upsilon}, we can rewrite $\Upsilon$ by
	\begin{align}
		\label{Upsilon_prefinal}
		\Upsilon =& \int_{y=0}^{\infty} \lambda_z e^{-\lambda_z z}~ e^{-\lambda_y \Lambda_S} {\text d}z - \frac{\lambda_y \lambda_z e^{-\lambda_x \mathcal{L} - \lambda_x \Lambda_S \mathcal{M}_\Upsilon - \lambda_y \Lambda_S }}{\lambda_x \mathcal{K}^{}_\Upsilon } \nonumber\\
		&\times   \int_{y=0}^{\infty} \frac{e^{-z \left(\lambda_z +\lambda_x \Lambda_S \mathcal{K}^{}_\Upsilon  \right)    } }{z + \frac{ \mathcal{M}^{}_\Upsilon}{ \mathcal{K}^{}_\Upsilon} + \frac{\lambda_y}{\lambda_x \mathcal{K}^{}_\Upsilon}}.
\end{align}		
Now, by using \cite[Eq. (3.352.4)]{gradshteyn2007}, the term $\Upsilon$ can be derived in a closed-form as
	\begin{align}
		\label{Upsilon_final}
		\Upsilon &= e^{-\lambda_y \Lambda_S} + \frac{\lambda_y \lambda_z e^{-\lambda_x \left(  \mathcal{L} + \Lambda_S \mathcal{M}_\Upsilon\right)  - \lambda_y \Lambda_S }}{\lambda_x \mathcal{K}^{}_\Upsilon }  e^{\mu^{}_S \xi^{}_S} \text{Ei} \left[- \mu^{}_S \xi^{}_S  \right], 
	\end{align}
	where $\mu^{}_S = \lambda_z +\lambda_x \Lambda_S \mathcal{K}^{}_\Upsilon$; $\xi^{}_S = \frac{ \mathcal{M}^{}_\Upsilon}{ \mathcal{K}^{}_\Upsilon} + \frac{\lambda_y}{\lambda_x \mathcal{K}^{}_\Upsilon}$ and $\text{Ei} [\cdot]$ is the exponential integral function.   

Finally, by inserting \eqref{Delta} and \eqref{Upsilon_final} into \eqref{CDF_gammaRj}, the closed-form expression for the CDF of $\gamma_{R,j}$ can be written as in \eqref{CDF_gammaRj_final}, where $\psi^{}_j <\frac{\alpha^{}_j}{\mathcal{A}}$, otherwise, $F_{\gamma^{}_{R,j}} (\psi^{}_j) \sim 1$. \qed
\section{Proof of Proposition 2}
\label{Appendix 2}
The term $\Theta$ in \eqref{CDF_gammabj} can be further extended as
{\allowdisplaybreaks \begin{align}
\label{Theta}
&\Theta = \int_{w=0}^{\infty} f_W(w) \int_{q=0}^{w  \mathcal{S}_\Theta + \mathcal{O}_\Theta + \mathcal{T}} \hspace{-1.2cm}f_{Q}(q) {\text d}q {\text d}w  \underbrace{\int_{v =0}^{\Lambda_R} f_V(v){\text d}v}_{\text{$\Theta_1$}} \nonumber\\
& = \lambda_w \int_{0}^{\infty} \hspace{-0.15cm} e^{-\lambda_w w} {\text d}w  - \lambda_w e^{-\lambda_q \left(\mathcal{O}_\Theta +\mathcal{T} \right)} \int_{0}^{\infty} \hspace{-0.15cm} e^{-w \left(\lambda_w + \lambda_q \mathcal{S}_\Theta \right)} {\text d}w \nonumber\\
& = \left(1 - e^{-\lambda_v \Lambda_R} \right) \left(1- \frac{\lambda_w e^{-\lambda_q \left(\mathcal{O}_\Theta + \mathcal{T} \right) }}{\lambda_w + \lambda_q \mathcal{S}_\Theta } \right).
\end{align}}Further, we extend the term $\Phi$ in \eqref{CDF_gammabj}	as
\begin{align}
\label{Phi_cont}
\Phi&= \int_{w=0}^{\infty} f_W(w) \underbrace{\int_{v=\Lambda_R}^{\infty} \int_{q=0}^{vw \mathcal{S}_\Phi + v \mathcal{O}_\Upsilon + \mathcal{L}}  f_V(v) f_Q(q){\text d}q {\text d}v}_{\Phi_1}  {\text d} w,
\end{align}
where $\Phi_1$ can be derived as
\begin{align}
\label{Phi_1}
\Phi_1 &= \lambda_v \int_{y=\Lambda_R}^{\infty} e^{-\lambda_v v} \left(1 - e^{-\lambda_q \left(w v \mathcal{S}_\Phi + v \mathcal{O}_\Phi + \mathcal{T} \right)} \right)  \text{d} v \nonumber\\
&=  e^{-\lambda_v \Lambda_R} - \frac{\lambda_v e^{-\lambda_q \mathcal{T}}~ e^{-\Lambda_R \left(w \lambda_q  \mathcal{S}_\Phi + \lambda_q \mathcal{O}_\Phi+\lambda_v \right) } }{w \lambda_q  \mathcal{S}_\Phi + \lambda_q \mathcal{O}_\Phi+\lambda_v }.
\end{align}

Then, by inserting \eqref{Phi_1} into \eqref{Phi_cont}, we can solve $\Phi$ as
\begin{align}
\label{Phi_final}
\Phi = e^{-\lambda_v \Lambda_R} + \frac{  e^{-\Lambda_R  \left( \lambda_v + \lambda_q \mathcal{O}_\Phi \right) - \lambda_q \mathcal{T} } }{ \left( \lambda_v \lambda_w \right)^{-1}  \lambda_q \mathcal{S}_\Phi} 
\left(e^{\xi^{}_R \mu^{}_R} \right) \text{Ei} \left[- \mu^{}_R \xi^{}_R \right], 
\end{align}		
where $\mu^{}_R = \lambda_w +\lambda_q \Lambda_R \mathcal{S}^{}_\Phi$ and $\xi^{}_R = \frac{ \mathcal{O}^{}_\Phi}{ \mathcal{S}^{}_\Phi} + \frac{\lambda_v}{\lambda_q \mathcal{S}^{}_\Phi}$. 	
Finally, after inserting \eqref{Theta} and \eqref{Phi_final} into \eqref{CDF_gammabj}, the closed-form solution for the CDF of $\gamma_{b,j}$ can be found as in \eqref{CDF_gammabj_final}, where $\psi^{}_j <\frac{\beta^{}_b}{\mathcal{J}_b}$, otherwise, $F_{\gamma^{}_{b,j}} (\psi^{}_j) \sim 1$.  \qed

\balance
\bibliographystyle{ieeetr}
\bibliography{mybibfile}
\end{document}